\documentclass[authoryear,times,3p]{elsarticle}

 \usepackage{amssymb}
 \usepackage{amsthm}
 \usepackage{amsmath}
 \usepackage{natbib}
 \usepackage{tikz}
 \usepackage{graphicx}
 \usepackage{epstopdf}
 \usepackage{xurl}
 \usepackage{hyperref}
 \usepackage{multirow}
 \usepackage{lineno}
 \theoremstyle{plain}
 \newdefinition{definition}{Definition}
\newtheorem{theorem}{Theorem}

  \newtheorem{proposition}{Proposition}
 \newdefinition{remark}{Remark}
 \newdefinition{property}{Property}
  \newdefinition{example}{Example}

\usepackage{bbm}

\makeatletter
\def\ps@pprintTitle{%
 \let\@oddhead\@empty
 \let\@evenhead\@empty
 \def\@oddfoot{\footnotesize\itshape
   \hfill \today}%
 \let\@evenfoot\@oddfoot}
\makeatother

\begin{document}
\setlength\bibsep{0pt}
\renewcommand*{\bibfont}{\footnotesize}

\begin{frontmatter}



\title{Spearman’s rho for zero-inflated count data: formulation and attainable bounds}


\author[TUe]{Jasper Arends}
\author[UWin]{Guanjie Lyu}
\author[UQ]{Mhamed Mesfioui}
\author[TUe]{Elisa Perrone*}
\author[ULB]{Julien Trufin}

\address[TUe]{Department of Mathematics and Computer Science, Eindhoven University of Technology, Groene Loper 5, 5612 AZ, Eindhoven, The Netherlands}
\address[UWin]{Department of Mathematics and Statistics, University of Windsor}
\address[UQ]{Department of Mathematics and Computer Science, University of Quebec in Trois-Riv\`eres, Canada}
\address[ULB]{Department of Mathematics, Universit\'e Libre de Bruxelles, Belgium}

\begin{abstract}
We propose an alternative formulation of Spearman’s rho for zero-inflated count data. The formulation yields an estimator with explicitly attainable bounds, facilitating interpretation in settings where the standard range $[-1,1]$ is no longer informative.
\end{abstract}


\end{frontmatter}


\section{Introduction}
Measures of association play an important role in quantifying dependence between random variables in many applications. 
Among the various measures proposed in the literature \citep{Blomqvist1950, Goodman1954, Nelsen2006, Scarsini1984}, Kendall’s tau ($\tau$) \citep{Kendall_1938} and Spearman’s rho ($\rho_S$) \citep{Spearman1904} are particularly popular due to their simple interpretation and rank-based estimation.
However, difficulties arise when the variables under study contain discrete components. Standard estimators of $\tau$ and $\rho_S$ rely on ranking the observations, and commonly used tie-breaking corrections may perform poorly when ties are substantial. 
Beyond estimation, discreteness also affects interpretation. For continuously distributed random variables, both $\tau$ and $\rho_S$ take values in the full interval $[-1,1]$, where the extremes correspond to perfect negative and positive association \citep{Nelsen2006}. 
In contrast, when the variables have discrete components, the attainable range of $\rho_S$ is generally a strict sub-interval of $[-1,1]$, depending on the (unknown) marginal distributions \citep{Denuit2002, Neslehova2007}.
As a result, the numerical value of $\rho_S$ alone may be misleading: a small positive value can still correspond to a substantial level of association if the attainable range is narrow. This highlights the importance of deriving attainable bounds in order to ensure a meaningful interpretation of any estimator of $\rho_S$.

In this paper, we address these issues in the context of zero-inflated data, that is, data generated by random variables with an increased probability mass at zero. Such data arise in many applications, including insurance, health care, ecology, and atmospheric sciences \citep{Arab2012,Mesfioui2022a, Moulton1995,Pechon2021}.
Recent contributions have highlighted the need for tailored formulations and estimators of rank-based association measures in the presence of zero inflation. In particular, \cite{Pimentel2015}, \cite{Denuit2017}, and \cite{Arends2025} studied concordance-based measures for zero-inflated continuous data and derived corresponding estimators and attainable ranges for Kendall’s tau, Spearman’s rho, among other measures. More recently, \cite{Perrone2023} extended these ideas to zero-inflated count data, with a focus on Kendall’s tau.
A different and more general perspective is provided by \cite{Nasri2023}, who proposed an estimator of Spearman’s rho, among other measures, applicable to discrete, continuous, and mixed data, including high-dimensional settings. While their approach offers a broadly applicable estimation framework, it does not explicitly address the derivation of attainable bounds in the presence of zero inflation.
Building on \cite{Arends2025}, we propose an alternative formulation of Spearman’s rho for zero-inflated count data. Although zero-inflated count data are in fact discrete data, the additional probability mass at zero provides a structure that enables explicit derivations for the bounds that are generally not available for arbitrary discrete distributions. Our derived formulation allows for a clearer characterization of the behavior of rank-based association measures, facilitates comparisons with Kendall's tau derived in \cite{Perrone2023}, and leads to an estimator whose achievable bounds can be explicitly derived, thus improving interpretation in settings where the standard range $[-1,1]$ is no longer informative.
The development of these results in the remainder of the paper is organized as follows. Section~\ref{sec:back} reviews Spearman’s rho and introduces the required notation. Section~\ref{sec:results} presents the new formulation and the corresponding estimator, and derives the attainable bounds. Simulation and case studies, together with empirical analysis of the statistical properties of the estimator are reported in Section~\ref{sec:sim}. Section~\ref{sec:concl} concludes with a discussion of possible extensions.

\section{Background and notation}
\label{sec:back}

We briefly recall the definition of Spearman’s rho and summarize the results relevant for the present work. Spearman’s rank correlation, denoted by $\rho_S$, was introduced by \cite{Spearman1904} as a rank-based alternative to Pearson’s correlation coefficient. Let $(X_1, Y_1)$, $(X_2, Y_2)$ and $(X_3, Y_3)$ be three independent copies of the pair $(X, Y)$ with joint distribution function $H$, then $\rho_S$ quantifies the association in terms of the probabilities of concordance and discordance as
\begin{equation}\label{eq:spm_rho}
    \rho_S = 3 P\bigl((X_1 - X_2)(Y_1 - Y_3) > 0\bigr) - 3 P\bigl((X_1 - X_2)(Y_1 - Y_3) < 0\bigr).
\end{equation}
Observe that the random variables $X_3$ and $Y_2$ do not appear explicitly in this expression. Even so, under dependence, they may still affect the distributions of $Y_3$ and $X_2$, respectively, and hence indirectly contribute to the value of $\rho_S$.
By rearranging the probabilities in Eq.~\eqref{eq:spm_rho}, we obtain
\begin{equation}\label{eq:spm_rho_discr}
    \rho_S = 6 P\bigl((X_1 - X_2)(Y_1 - Y_3) > 0\bigr) - 3 P\bigl(X_1 = X_2 \text{ or } Y_1 = Y_3\bigr) - 3.
\end{equation}
If $(X, Y)$ is a continuous random vector, then the probability of ties is zero and only the remaining component affects $\rho_S$. However, as shown by Eq.~\eqref{eq:spm_rho_discr}, $\rho_S$ is heavily affected by the probability of ties if this is non-negligible, as in the case of zero-inflated data.

We introduce notation for zero-inflated random variables. Following \cite{Pimentel2015,Perrone2023,Arends2025}, let $X$ and $Y$ be two zero-inflated random variables with increased probability masses $p_1$ and $p_2$ at zero. Their cumulative distribution functions can then be written as $F(x) = p_1 + (1 - p_1) \tilde F(x)$, and $G(y) = p_2 + (1 - p_2) \tilde G(y)$, for all $x, y \geq 0$. Here, $\tilde F$ and $\tilde G$ are the original (not extra inflated) distribution functions. If $\tilde F$ and $\tilde G$ are continuous, then the probability mass at zero is only given by the extra mass $p_1$ and $p_2$. The focus of this paper, however, lies on zero-inflated discrete distributions, that is when $\tilde F$ and $\tilde G$ are discrete and have a strictly positive support in zero. Therefore, the probability mass around zero in such a case is $P(X = 0) = p_1 + (1 - p_1) \tilde F(0)$ and $P(Y = 0) = p_2 + (1 - p_2) \tilde G(0)$.

The idea underlying the proposed estimator is to separate contributions arising from observations at zero and from strictly positive observations. 
To this end, we consider the probability $p_{ab} = P\bigl(\mathbbm 1(X > 0) = a, \mathbbm 1(Y > 0) = b\bigr)$ for all $a, b \in \{0, 1\}$. Moreover, let $X_{10}$ ($X_{11}$) be the positive random variable with the distribution of $X$, given that $Y = 0$ ($Y > 0$), and define $Y_{01}$ and $Y_{11}$ analogously. Since the distributions of these random variables vary for different relations between $X$ and $Y$, we need to also consider the probabilities $p_1^* = P(X_{10} > X_{11})$, $p_1^\dagger = P(X_{10} = X_{11})$, $p_2^* = P(Y_{01} > Y_{11})$ and $p_2^\dagger = P(Y_{01} = Y_{11})$. 

We next recall existing results on Spearman's rho for zero-inflated continuous distributions. Let $S^* = \{X_1 > 0, Y_1 > 0, X_2 > 0, Y_3 > 0\}$. For zero-inflated continuous distributions, \cite{Arends2025} showed that Spearman's rho can be expressed by
\begin{equation}\label{eq:spm_czid}
    \rho_S = p_{11} (1 - p_1)(1 - p_2) \rho_{S^*} + 3 p_{11} \bigl(p_{10} (1 - 2p_1^*) + p_{01} (1 - 2p_2^*)\bigr) + 3 (p_{00} p_{11} - p_{01} p_{10}),
\end{equation}
where $\rho_{S^*}$ is the difference between the conditional probabilities of concordance and discordance given $S^*$, that is
\begin{align}\label{eq:spm_rho*}
    \rho_{S^*} = 3 P\bigl((X_1 - X_2)(Y_1 - Y_3) > 0\, |\, X_1 > 0, X_2, Y_1, Y_3 > 0 \bigr) - 3 P\bigl((X_1 - X_2)(Y_1 - Y_3) < 0\, |\, X_1, X_2, Y_1, Y_3 > 0\bigr).
\end{align}
\cite{Denuit2017} examined two dependence structures where the supports of $X_{11}$ and $X_{10}$ (similarly for $Y_{11}$ and $Y_{01}$) are disjoint. This highlights the importance of distinguishing the cases where $X_3$ and $Y_2$ are zero or strictly positive as well, despite their absence in Eq.~\eqref{eq:spm_rho*}. This led \cite{Arends2025} to the formulation of $\rho_{S^*}$ as
\begin{align}\label{eq:rhoS_pos}
    (1 - p_1)(1 - p_2) \rho_{S^*} = p_{11}^2 \rho_{S_{11}} + p_{11} p_{10} \rho_{S_{10}} + p_{11} p_{01} \rho_{S_{01}} + p_{10} p_{01} \rho_{S_{00}},
\end{align}
where $\rho_{S_{ab}}$ is Spearman's rho conditional on $S_{ab} = S^* \cap \{\mathbbm 1 (X_3 > 0) = a, \mathbbm 1(Y_2 > 0) = b\}$ for $a, b \in \{0, 1\}$. For example, taking $a = 0$ and $b = 0$ here yields $S_{00} = \{X_1 > 0, Y_1 > 0, X_2 > 0, Y_2 = 0, X_3 = 0, Y_3 > 0\}$, and hence $X_2$, $Y_3$ are distributed as $X_{10}$, $Y_{01}$ respectively. Based on this representation, they proposed an estimator of Spearman's rho by replacing each term in Eq.~\eqref{eq:spm_czid} and Eq.~\eqref{eq:rhoS_pos} by their relative frequencies and estimating $\rho_{S_{11}}$ using the classical Spearman's rho computed from observations with both $X>0$ and $Y>0$. In the next section, we first derive a new estimator for Spearman's rho for zero-inflated count data, and then its attainable bounds.

\section{Estimator and attainable bounds}
\label{sec:results}

\subsection{Estimator of Spearman's rho for zero-inflated count data}
\label{sec:est}

The derivation of the new estimator follows the general approach adopted in \cite{Pimentel2015,Perrone2023,Arends2025}, whereby the association measure is decomposed into contributions arising from observations at zero and from strictly positive observations. The resulting expression is stated in the following theorem.
\begin{theorem}\label{thm:spm_dzid}
    For bivariate zero-inflated count data, Spearman's rho is given by
    \begin{equation}\label{eq:spm_dzid}
        \rho_S = p_{11} (1 - p_1)(1 - p_2)\rho_{S^*} + 3 p_{11} \Bigl(p_{10} \bigl(1 - 2 p_1^* - p_1^\dagger\bigr) + p_{01} \bigl(1 - 2 p_2^* - p_2^\dagger\bigr)\Bigr) + 3 (p_{00} p_{11} - p_{01} p_{10}).
    \end{equation}
\end{theorem}
\begin{proof}
    The proof is delegated to the Appendix.
\end{proof}

The expression reported in Eq.~\eqref{eq:spm_dzid} applies to both zero-inflated continuous and zero-inflated discrete distributions. 
In the continuous case, it reduces to expression obtained in \cite{Arends2025} reported in Eq.~\eqref{eq:spm_czid}, since the probabilities of ties away from zero, $p_1^\dagger$ and $p_2^\dagger$, vanish.

An estimator of Spearman’s rho can be constructed by replacing each component in Eq.~\eqref{eq:spm_dzid} with its empirical counterpart. In analogy with \cite{Arends2025}, $\rho_{S_{11}}$ is estimated using the classical Spearman’s rho computed from observations with $X>0$ and $Y>0$, with ties handled by average ranks. The resulting estimator is denoted by $\hat{\rho}_A$. 
The components estimated by relative frequencies are maximum likelihood estimates (MLE) and hence satisfy standard asymptotic properties. However, the proposed estimator of Spearman's rho based on Eq.~\eqref{eq:spm_dzid} is not itself a full MLE, because the component $\rho_{S_{11}}$, i.e., Spearman's rho for the strictly positive observations, is estimated by the standard rank-based approach, namely Pearson correlation of the average ranks in the presence of ties. 
When the probabilities of inflation are large, the contribution of $\hat \rho_{S_{11}}$ to the overall estimator becomes limited, so asymptotic normality of the full estimator remains plausible. Due to the complicated dependence structure between $\hat \rho_{S_{11}}$, $\hat \rho_{S_{10}}$, and the other estimated components, we use bootstrap to assess the estimator's asymptotic behavior. We discuss this aspect in Section~\ref{sec:sim}.
For completeness, we note that a more general estimator of $\rho_S$ applicable to arbitrary distributions has been proposed by \cite{Nasri2023}. A numerical comparison between the two estimators is also reported in Section~\ref{sec:sim}.

\subsection{Attainable bounds for Spearman's rho}
\label{sec:bounds}

For continuous distributions, $\rho_S$ takes values in the interval $[-1,1]$, with the extremes corresponding to the strongest negative and positive association \citep{Nelsen2006}. When the variables include discrete components, however, the attainable range of $\rho_S$ is generally a strict subinterval of $[-1,1]$ and depends on the marginal distributions \citep{Neslehova2007}. As a consequence, the numerical value of $\rho_S$ alone may be misleading, and an explicit characterization of its attainable bounds is required for meaningful interpretation.

For general discrete random variables, sharp bounds for $\rho_S$ were derived by \cite{Mesfioui2022c}, but their expressions involve several expectations depending on the marginal distributions, which limits their practical applicability. Here, we derive more explicit bounds tailored to zero-inflated count data by exploiting the formulation in Eq.~\eqref{eq:spm_dzid}.
Sharp bounds for $\rho_S$ and $\tau$ in the zero-inflated continuous setting were obtained by \cite{Denuit2017} and \cite{Arends2025}. The zero-inflated discrete case of $\tau$ was addressed in \cite{Perrone2023}. These results rely on the theory of copulas and the fact that extremal values are attained under the Fréchet--Hoeffding copula bounds. By Sklar’s theorem \citep{Sklar1959}, any joint distribution with margins $F$ and $G$ can be written as $H(x,y)=C(F(x),G(y))$, where $C$ is a copula \citep{Nelsen2006}. The attainable bounds of $\rho_S$ are therefore obtained by evaluating Eq.~\eqref{eq:spm_dzid} under the Fréchet--Hoeffding copulas given by $M(u,v)=\min\{u,v\}$ and $W(u,v)=\max\{u+v-1,0\}$.

Before stating the bounds, we introduce additional notation. For the upper bound, let $\tilde s, \tilde u, \tilde t, \tilde v \in \mathbb N$ be the points such that
\begin{alignat*}{3}
    &G(\tilde t-1) \leq p_1 < G(\tilde t), &&F(\tilde v-1) < G(\tilde t) \leq F(\tilde v),\qquad && \text{if } p_1 \leq p_2, \\
    &F(\tilde s-1) \leq p_2 < F(\tilde s),\qquad &&G(\tilde u-1) < F(\tilde s) \leq G(\tilde u),\qquad && \text{if } p_1 > p_2.
\end{alignat*}
We recall from \cite{Perrone2023} that $\tilde s$ is the point that is attainable by both $X_{10}$ and $X_{11}$, similarly for $\tilde t$ with respect to $Y_{01}$ and $Y_{11}$. For the lower bound, define $\tilde s', \tilde u', \tilde t', \tilde v' \in \mathbb N$ to be the points that satisfy
\begin{alignat*}{2}
    &F(\tilde s'-1) \leq 1 - p_2 < F(\tilde s'), \qquad && G(\tilde u'-1) \leq 1 - F(\tilde s'-1) < G(\tilde u'), \\
    &G(\tilde t'-1) \leq 1 - p_1 < G(\tilde t'), \qquad && F(\tilde v'-1) \leq 1 - G(\tilde t'-1) < F(\tilde v'),
\end{alignat*}
when $p_1 + p_2 < 1$. Additionally, let $I_1 = \min\{1 - p_2, F(\tilde v')\}$ and $I_2 = \min\{1 - p_1, G(\tilde u')\}$. Now, the sharp upper and lower bounds on $\rho_S$ can be formulated as follows.

\begin{proposition}\label{prop:bounds}
    Let $\rho_{S_{11}}^{\max}$ and $\rho_{S_{11}}^{\min}$ be sharp upper and lower bounds on Spearman's rho on the random variables that are positive. Then, the sharp upper and lower bounds on $\rho_S$ for a pair of zero-inflated discrete random variables are given by
    \begin{equation*}
        \rho_S^{\max} = \begin{cases}
            (1 - p_1)^3 \rho_{S_{11}}^{\max} + 3 p_1 (1 - p_1) + 3 \bigl(p_1 - G(\tilde t-1)\bigr) \Bigl(G(\tilde t) \bigl(p_1 - F(\tilde v) - F(\tilde v-1)\bigr) + F(\tilde v) F(\tilde v-1)\Bigr), & \text{if } p_1 \leq p_2, \\
            (1 - p_2)^3 \rho_{S_{11}}^{\max} + 3 p_2 (1 - p_2) + 3 \bigl(p_2 - F(\tilde s-1)\bigr) \Bigl(F(\tilde s) \bigl(p_2 - G(\tilde u) - G(\tilde u-1)\bigr) + G(\tilde u) G(\tilde u-1)\Bigr), &\text{if } p_1 > p_2,
        \end{cases}
    \end{equation*}
    and
    \begin{equation*}
        \rho_S^{\min} = \begin{cases}
            -3 (1 - p_1)(1 - p_2), \quad \text{if } p_1 + p_2 \geq 1, \\
            w(0, 0)^3 \rho_{S_{11}}^{\min} + 3 w(0, 0) (p_1 p_2 - p_1 - p_2) - 3 p_1 p_2 + w(\tilde s'-1, \tilde u'-1) w(\tilde s', 0) w(0, \tilde t') \mathbbm 1(\tilde s' = \tilde v') + \\
            \quad - 3 w(\tilde s', 0) \left[p_2 w(\tilde s'-1, 0) + \bigl(G(\tilde u'-1) - p_2\bigr)^2 + w(\tilde s'-1, \tilde u'-1) \bigl(I_2 - 2 p_2 + G(\tilde u'-1)\bigr)\right] \\
            \quad - 3 w(0, \tilde t') \left[p_1 w(0, \tilde t'-1) + \bigl(F(\tilde v'-1) - p_1\bigr)^2 + w(\tilde v'-1, \tilde t'-1) \bigl(I_1 - 2 p_1 + F(\tilde v'-1)\bigr)\right], \quad \text{if } p_1 + p_2 < 1,
        \end{cases}
    \end{equation*}
    where $w(\cdot, \cdot)$ is the function defined by $w(x, y) = 1 - F(x) - G(y)$ for $x, y \in \mathbb N$.
\end{proposition}

The bounds given here are consistent with the ones for zero-inflated continuous settings studied by \cite{Arends2025}. In particular, the bounds can be simplified to $\rho_S^{\max} = 1 - \max\{p_1, p_2\}^3$ and $\rho_S^{\min} = p_1^3 + p_2^3 - 1$ if $p_1 + p_2 \geq 1$, since for $\tilde F$ and $\tilde G$ continuous, one obtains $G(\tilde t) = p_1$, $F(\tilde s) = p_2$ for the upper bound and $G(\tilde t') = G(\tilde t'-1) = 1 - p_1$, $F(\tilde s') = F(\tilde s'-1) = 1 - p_2$ and in turn $G(\tilde u') = G(\tilde u'-1) = p_2$ and $F(\tilde v') = F(\tilde v'-1) = p_1$.

The bounds in Proposition~\ref{prop:bounds} can be estimated using the empirical distributions of $X$ and $Y$ and using the relaxation $-1 \leq \rho_{S_{11}} \leq 1$. We will briefly investigate the performance of these estimators in a simulation study in Section~\ref{sec:sim}. 

\section{Simulation study and empirical application}
\label{sec:sim}

\subsection{Simulation study}
We now perform a simulation study to evaluate the effectiveness of the proposed theory, and compare our estimator, $\hat \rho_A$, to that of \cite{Nasri2023}, $\hat \rho_N$. Here, we report the analysis of the zero-inflated discrete distributions based on a similar setting studied by \cite{Perrone2023}. We compute the values of the estimators of $\rho_S$ for $N$ pairs generated from two random variables joined through the Fr\'echet copula $C_\alpha(u, v) = (1 - \alpha) u v + \alpha \min\{u, v\}$, where $u, v, \alpha \in [0, 1]$ \citep{Nelsen2006}. We compare the derived estimates $\hat \rho_A$ and $\hat \rho_N$ with the true value of $\rho_S$ calculated as proposed by \cite{Safari2020}, who expressed $\rho_S$ in terms of the margins and joint distribution. An implementation of their formula is made available by \cite{Greef2020}. 
As the estimator varies for different margins, we selected multiple scenarios representative of various characteristics of the samples. 

Namely, we illustrate the performance of our estimator for Poisson margins, with parameters $\lambda_F, \lambda_G$. To assess the finite-sample performance and empirical statistical properties of the estimator, we consider $p_1,p_2 \in \{0.2,0.8\}$ for the zero-inflation probabilities, $\lambda_F,\lambda_G \in \{2,8\}$ for the Poisson margins, and $\alpha \in \{0.2,0.5,0.8\}$ for the dependence level. Smaller values of $\lambda_F$ and $\lambda_G$ induce more discreteness away from zero, while for value $8$ this effect is nearly negligible. We also vary the sample size over $N \in \{100,150,250,500,1000\}$ to study the estimator’s behavior in larger samples. Finally, for $N=150$, we use 1000 bootstrap resamples to examine the estimation of its standard deviation. An implementation of the simulation study in \texttt{R} \citep{RCore_2025} is available on GitHub at \url{https://github.com/JasperArends/SpmZICD}.

Figure~\ref{fig:simulation} presents boxplots for a representative subset of the simulation results. The complete numerical results are reported in the supplementary material. In particular, Table~S.1 reports the mean, the estimated mean squared error multipled by a factor of 100 (MSE*), and the estimated standard deviation across the simulations for the two estimators of $\rho_S$. Overall, our estimator performs comparably to the estimator proposed by \cite{Nasri2023}, which supports Theorem~\ref{thm:spm_dzid}. Moreover, the estimated standard deviations decrease substantially as the sample size increases.
Table~S.2 further shows that the bootstrap estimate of the standard deviation closely matches the corresponding Monte Carlo estimate on average. In addition, the QQ-plots in Figure~S.1 suggest that a normal approximation for the estimator is generally reasonable.
Finally, Table~S.3 reports the simulation results for the estimation of the attainable bounds. The bounds are generally estimated accurately across the considered parameter configurations. Small deviations may occur when the level of zero-inflation is low, which can be attributed to the relaxation $-1 \leq \rho_{S_{11}} \leq 1$ and to the non-differentiability of the bounds with respect to the margins, for example through constraints such as $p_{11} \leq 1 - \min\{p_1,p_2\}$.
Table~S.3 also illustrates how the attainable bounds influence the interpretation of the estimates. For instance, when $p_1=p_2=0.8$ and $\lambda_F=\lambda_G=2$ with dependence level $\alpha=0.8$, the mean estimate of Spearman’s rho equals $0.35$, while the corresponding upper bound is $0.40$. Relative to this attainable bound, the estimated value indicates a strong level of association.

\begin{figure}[t]
    \centering
    \includegraphics[width=.22\linewidth]{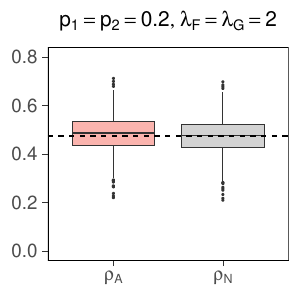}
    \includegraphics[width=.22\linewidth]{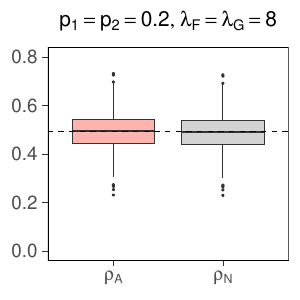}
    \includegraphics[width=.22\linewidth]{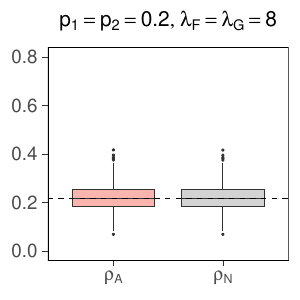}
    \includegraphics[width=.22\linewidth]{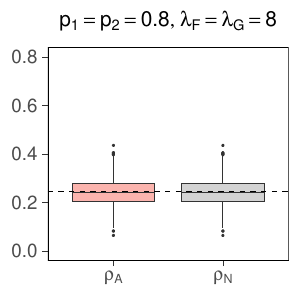} \vspace{-.5cm}
    \caption{Simulation results for $\alpha = 0.5$ and $N = 150$ samples, based on 1000 repetitions.}
    \label{fig:simulation}
\end{figure}

\subsection{Case study}
In this section, we illustrate how the proposed theory can be applied to real data with zero-inflated count variables. To this end, we consider the insurance data set of \cite{Guillen2021}, which contains information on homeowner and motor insurance claims. The data tracks 40,284 policyholders in Spain over the period from 2010 until 2014. For each customer, the numbers of claims are aggregated over time, yielding the total number of claims for each insurance type. Zero-inflation arises naturally because most customers do not experience an accident and therefore file no claim. In particular, the estimated probabilities of observing zero claims equal $\hat p_M = 0.94$ for motor insurance and $\hat p_H = 0.91$ for homeowner insurance.

We estimate Spearman's rho using both the estimator proposed in Section~\ref{sec:est} and the estimator of \cite{Nasri2023}, and both yield the same value, namely $\hat\rho_S = 0.0031$. Based on the estimated zero-inflation probabilities and the empirical margins, the attainable bounds are estimated as $-0.016 \leq \hat\rho_S \leq 0.17$. Relative to these bounds, the level of association is somewhat stronger than the raw estimate alone may suggest, although it remains small in absolute terms.

\section{Conclusion}
\label{sec:concl}
In this paper, we derived an alternative formulation of Spearman’s rho for bivariate zero-inflated count data and obtained explicit attainable bounds. The proposed estimator and its bounds are easily computed from the data, allowing for a meaningful interpretation of the strength of association when the standard range $[-1,1]$ is no longer informative.
In future research, we plan to use the new expression of $\rho_S$, and its counterpart for $\tau$ from \cite{Pimentel2015,Perrone2023}, to analyze the relationship between the two association measures for zero-inflated data, thereby extending known results for the continuous case \citep{Schreyer2017}. This was already briefly investigated by \cite{Arends2025b} for zero-inflated continuous distributions. Moreover, \cite{Mesfioui2022b} established bounds on a multivariate extension of Spearman's rho in case of zero-inflated continuous distributions, however these are not sharp. It would be of interest to investigate whether the approach we propose here can be adapted to derive sharper bounds for multivariate extensions of Spearman’s rho. \newline

\noindent \textbf{Supplementary material} \\
In the supplementary material, we provide the complete proofs of Theorem~\ref{thm:spm_dzid} and Proposition~\ref{prop:bounds}, alongside with the tables and plots discussed in Section~\ref{sec:sim}.

\bibliography{References_SPL}

\section*{Supplementary material}

\setcounter{section}{0}
\setcounter{table}{0}
\setcounter{figure}{0}

\renewcommand{\thesection}{S.\arabic{section}}
\renewcommand\thetable{S.\arabic{table}}
\renewcommand\thefigure{S.\arabic{figure}}

In the following, we present the supplementary material, including the complete proofs of Theorem 1 and Proposition 1, and the results of the simulation study discussed in Section 4 of the manuscript.

\section{Proof of Theorem 1}
Here, we present the complete proof of Theorem~1. The idea is to separate all cases where $X$ and $Y$ are zero or strictly positive. To that end, let
\begin{align*}
    C_{a_1 b_1 a_2 b_2} = \bigl\{&\mathbbm 1(X_1 > 0) = a_1, \mathbbm 1(Y_1 > 0) = b_1, \mathbbm 1(X_2 > 0) = a_2, \mathbbm 1(Y_3 > 0) = b_2\bigr\},
\end{align*}
for $a_1, b_1, a_2, b_2 \in \{0, 1\}$, and observe that $S^* = C_{1111}$. To simplify notation, we define $p_{0+} = P(X = 0)$, $p_{1+} = P(X > 0)$, $p_{+0} = P(Y = 0)$ and $p_{+1} = P(Y > 0)$. We note that it is possible to write these terms with respect to $p_1$ and $p_2$. However, this notation highlights the dependency on each variable more. We consider $\rho_S$ written as follows
\begin{align*}
    \rho_S = 3 \bigl(P(X_1 < X_2, Y_1 < Y_3) + P(X_1 > X_2, Y_1 > Y_3) - P(X_1 < X_2, Y_1 > Y_3) - P(X_1 > X_2, Y_1 < Y_3)\bigr).
\end{align*}
We now analyze each of these terms individually and write them as
\begin{align*}
    P(X_1 < X_2, Y_1 < Y_3) &= P(X_1 < X_2, Y_1 < Y_3 \, |\,  C_{1111}) p_{11} p_{1+} p_{+1} \\
    &\qquad+ P(Y_1 < Y_3 \, |\,  C_{0111}) p_{01} p_{1+} p_{+1} \\
    &\qquad+ P(X_1 < X_2 \, |\,  C_{1011}) p_{10} p_{1+} p_{+1} \\
    &\qquad+ p_{00} p_{1+} p_{+1}, \\
    P(X_1 > X_2, Y_1 > Y_3) &= P(X_1 > X_2, Y_1 > Y_3 \, |\,  C_{1111}) p_{11} p_{1+} p_{+1} \\
    &\qquad+ P(Y_1 > Y_3 \, |\,  C_{1101}) p_{11} p_{0+} p_{+1} \\
    &\qquad+ P(X_1 > X_2 \, |\,  C_{1110}) p_{11} p_{1+} p_{+0} \\
    &\qquad+ p_{11} p_{0+} p_{+0}, 
\end{align*}
\begin{align*}    
    P(X_1 < X_2, Y_1 > Y_3) &=  P(X_1 < X_2, Y_1 > Y_3 \, |\,  C_{1111}) p_{11} p_{1+} p_{+1} \\
    &\qquad+ P(Y_1 > Y_3 \, |\,  C_{0111}) p_{01} p_{1+} p_{+1} \\
    &\qquad+ P(X_1 < X_2 \, |\,  C_{1110}) p_{11} p_{1+} p_{+0} \\
    &\qquad+ p_{01} p_{1+} p_{+0}, \\
    P(X_1 > X_2, Y_1 < Y_3) &= P(X_1 > X_2, Y_1 < Y_3 \, |\,  C_{1111}) p_{11} p_{1+} p_{+1} \\
    &\qquad+ P(Y_1 < Y_3 \, |\,  C_{1101}) p_{11} p_{0+} p_{+1} \\
    &\qquad+ P(X_1 > X_2 \, |\,  C_{1011}) p_{10} p_{1+} p_{+1} \\
    &\qquad+ p_{10} p_{0+} p_{+1}.
\end{align*}
We collect the last terms that can all be expressed in terms of $p_{11}$, $p_{10}$, $p_{01}$, $p_{00}$, $p_1$ and $p_2$. \cite{Pimentel2009} used the margins of $p_{0+} = p_{01} + p_{00}$, $p_{+0} = p_{10} + p_{00}$, $p_{1+} = p_{11} + p_{10}$ and $p_{+1} = p_{11} + p_{01}$ to rewrite them as follows
\begin{align*}
        p_{00} p_{1+}  & p_{+1} + p_{11} p_{0+} p_{+0} - p_{01} p_{1+} p_{+0} - p_{10} p_{0+} p_{+1} \\
    &= p_{00} (p_{10} + p_{11}) (p_{01} + p_{11}) + p_{11} (p_{00} + p_{01}) (p_{00} + p_{10}) \\
    &\qquad - p_{01} (p_{10} + p_{11}) (p_{00} + p_{10}) - p_{10} (p_{00} + p_{01}) (p_{01} + p_{11}) \\
    &= p_{00} (p_{10} p_{01} + p_{10} p_{11} + p_{01} p_{11} + p_{11}^2) + p_{11} (p_{00}^2 p_{10} p_{00} + p_{10} p_{00} + p_{01} p_{00} + p_{01} p_{10}) \\
    &\qquad - p_{01} (p_{10} p_{00} + p_{10}^2 + p_{11} p_{00} + p_{10} p_{11}) - p_{10} (p_{01} p_{00} + p_{00} p_{11} + p_{01}^2 + p_{01} p_{11}) \\
    &= p_{00} p_{11} (p_{10} + p_{01} + p_{11} + p_{00}) - p_{10} p_{01} (p_{10} + p_{00} + p_{01} + p_{11})^2 \\
    &= p_{00} p_{11} - p_{10} p_{01}.
\end{align*}
Substitution into the definition of $\rho_S$ yields $\rho_S = p_{11} p_{1+} p_{+1} \rho_{S^*} + 3A + 3(p_{00} p_{11} - p_{01} p_{10})$, where
\begin{align*}
    A =& P(Y_1 < Y_3\,|\,C_{0111}) p_{01} p_{1+} p_{+1} + P(X_1 < X_2\,|\,C_{1011}) p_{10} p_{1+} p_{+1} \\
    &+ P(Y_1 > Y_3\,|\,C_{1101}) p_{11} p_{0+} p_{+1} + P(X_1 > X_2\,|\,C_{1110}) p_{11} p_{1+} p_{+0} \\
    &- P(Y_1 > Y_3\,|\,C_{0111}) p_{01} p_{1+} p_{+1} - P(X_1 < X_2\,|\,C_{1110}) p_{11} p_{1+} p_{+0} \\
    &- P(Y_1 < Y_3\,|\,C_{1101}) p_{11} p_{0+} p_{+1} - P(X_1 > X_2\,|\,C_{1011}) p_{10} p_{1+} p_{+1} \\
    \equiv& A_Y + A_X.
\end{align*}
Here, $A_Y$ consists of all the terms involving $Y_1$, $Y_2$ and $Y_3$, whereas the remaining terms form $A_X$. We now only elaborate upon $A_Y$, the same procedure can be used to simplify $A_X$. Let us first expand upon each term separately and obtain
\begin{align*}
    P(Y_1 < Y_3\,|\,C_{0111}) p_{01} p_{1+} p_{+1} = P(Y_{01} < \tilde Y_{01}) \frac{p_{01}}{p_{+1}} p_{01} p_{1+} p_{+1} + P(Y_{01} < Y_{11}) \frac{p_{11}}{p_{+1}} p_{01} p_{1+} p_{+1}, \\
    P(Y_1 > Y_3\,|\,C_{1101}) p_{11} p_{0+} p_{+1} = P(Y_{11} > Y_{01}) \frac{p_{01}}{p_{+1}} p_{11} p_{0+} p_{+1} + P(Y_{11} > \tilde Y_{11}) \frac{p_{11}}{p_{+1}} p_{11} p_{0+} p_{+1}, \\
    P(Y_1 > Y_3\,|\,C_{0111}) p_{01} p_{1+} p_{+1} = P(Y_{01} > \tilde Y_{01}) \frac{p_{01}}{p_{+1}} p_{01} p_{1+} p_{+1} + P(Y_{01} > Y_{11}) \frac{p_{11}}{p_{+1}} p_{01} p_{1+} p_{+1}, \\
    P(Y_1 < Y_3\,|\,C_{1101}) p_{11} p_{0+} p_{+1} = P(Y_{11} < Y_{01}) \frac{p_{01}}{p_{+1}} p_{11} p_{0+} p_{+1} + P(Y_{11} < \tilde Y_{11}) \frac{p_{11}}{p_{+1}} p_{11} p_{0+} p_{+1}.
\end{align*}
Here, $\tilde Y_{01}, \tilde Y_{11}$ are independent copies of $Y_{01}, Y_{11}$ respectively. These alternative formulations can be substituted back into $A_Y$, resulting in several terms canceling out. In particular, we find that
\begin{align*}
    A_Y &= P(Y_1 < Y_3\,|\,C_{0111}) p_{01} p_{1+} p_{+1} + P(Y_1 > Y_3\,|\,C_{1101}) p_{11} p_{0+} p_{+1} + \\
    &\quad - P(Y_1 > Y_3\,|\,C_{0111}) p_{01} p_{1+} p_{+1} - P(Y_1 < Y_3\,|\,C_{1101}) p_{11} p_{0+} p_{+1} \\
    &= P(Y_{01} < Y_{11}) p_{11} p_{01} (p_{0+} + p_{1+}) - P(Y_{01} > Y_{11}) p_{11} p_{01} (p_{1+} + p_{0+}) \\
    &= p_{11} p_{01} \bigl(1 - 2 P(Y_{01} > Y_{11}) - P(Y_{01} = Y_{11})\bigr).
\end{align*}
Following the same steps for $A_X$ yields
\begin{align*}
    A = p_{11} \Bigl(p_{10} \bigl(1 - 2 p_1^* - p_1^\dagger\bigr) + p_{01} \bigl(1 - 2 p_2^* - p_2^\dagger\bigr) \Bigr).
\end{align*}
Finally, we focus on $p_{1+} p_{+1} \rho_{S^*}$. Conditioning on $S_{ab}$ for all possible combinations of $a, b \in \{0, 1\}$ results in
\begin{align*}
    p_{1+} p_{+1} \rho_{S^*} = p_{11}^2 \rho_{S_{11}} + p_{11} p_{10} \rho_{S_{10}} + p_{11} p_{01} \rho_{S_{01}} + p_{10} p_{01} \rho_{S_{00}},
\end{align*}
which completes the proof. \qed

\section{Proof of Proposition 1}
Here, we present a proof of Proposition~1, where the attainable bounds on $\rho_S$ for zero-inflated count variables are given. Recall from \cite{Mesfioui2005}, that the upper and lower bounds are attained under the Fr\'echet-Hoeffding copula bounds $M$ and $W$, i.e., when the joint distribution of $X, Y$ coincides with
\begin{align*}
    P(X \leq x, Y \leq y) = M(F(x), G(y)) = \min\{F(x), G(y)\},
\end{align*}
and
\begin{align*}
    P(X \leq x, Y \leq y) = W(F(x), G(y)) = \max\{0, F(x) + G(y) - 1\}
\end{align*}
respectively. \cite{Perrone2023} has shown that the two random variables $X_{10}$ and $X_{11}$ have a common support, and therefore $X_{10} > X_{11}$ and $X_{10} < X_{11}$ are true or false only with some probability. Therefore, it is important to have a detailed formulation of their distributions at hand. The derivations for the attainable bounds are set as follows: for the upper, as well as the lower bounds, we first derive (or recall from \cite{Perrone2023}) the necessary distribution functions from $X_{11}, X_{10}, Y_{11}$ and $Y_{01}$, and then compute the components of Eq.~(6) individually.

\subsection{Upper bound}
We consider the upper Fr\'echet-Hoeffding bound $M(u, v) = \min\{u, v\}$ for all $u, v \in [0, 1]$, and, without loss of generality, assume that $p_1 \leq p_2$. \cite{Denuit2017} noted that $p_{11} = 1 - p_2$, $p_{00} = p_1$, $p_{10} = p_2 - p_1$ and $p_{01} = 0$. The probability mass functions of $X_{11}$ and $X_{10}$ have already been derived by \cite{Perrone2023} and can be given by
\begin{equation*}
    P(X_{11} = x) = \begin{cases}
        0 & \text{if } 0 < x < \tilde s, \\
        \frac{F(x) - p_2}{1 - p_2} &\text{if } x = \tilde s, \\
        \frac{F(x) - F(x - 1)}{1 - p_2} & \text{otherwise},
    \end{cases} \qquad \text{and} \qquad P(X_{10} = x) = \begin{cases}
        \frac{F(x) - F(x - 1)}{p_2 - p_1} & \text{if } 0 < x < \tilde s, \\
        \frac{p_2 - F(x - 1)}{p_2 - p_1} & \text{if } x = \tilde s, \\
        0 & \text{otherwise,}
    \end{cases}
\end{equation*}
where $\tilde s \in \mathbb N$ is the point where $F(\tilde s - 1) \leq p_2 < F(\tilde s)$. The random variables $X_{10}$ and $X_{11}$ only have $\tilde s$ in their common support, and therefore $P(X_{11} \geq X_{10}) = 1$ and $P(X_{11} \leq X_{10}) = P(X_{11} = \tilde s) P(X_{10} = \tilde s)$. Clearly, the point $\tilde s$ plays a significant role in the distributions of $X$ and $Y$. We now derive several properties of this joint distribution revolving around $\tilde s$, that will be put to use later in the computation of the terms in Eq.~(6). First of all, it is clear that $\tilde s > 0$, then for $y > 0$,
\begin{equation*}\begin{split}
    P(X = \tilde s, Y &\leq y \ |\ X > 0, Y > 0) = \frac{P(X = \tilde s, 0 < Y \leq y)}{P(X > 0, Y > 0)} \\
    &= \frac{1}{1 - p_2} \Bigl(P(X = \tilde s, Y \leq y) - P(X = \tilde s, Y = 0)\Bigr) \\
    &= \frac{1}{1 - p_2} \Bigl(P(X \leq \tilde s, Y \leq y) - P(X \leq \tilde s  - 1, Y \leq y) + \\
    &\hspace*{2.1cm} - P(X \leq \tilde s, Y = 0) + P(X \leq \tilde s - 1, Y = 0)\Bigr) \\
    &= \frac{1}{1 - p_2} \Bigl(P(X \leq \tilde s, Y \leq y) - \min\{F(\tilde s - 1), G(y)\} + \\
    &\hspace*{2.7cm} - \min\{F(\tilde s), p_2\} + \min\{F(\tilde s - 1), p_2\}\Bigr) \\
    &= \frac{1}{1 - p_2} \Bigl(\min\{F(\tilde s), G(y)\} - p_2\Bigr).
\end{split}\end{equation*}
Therefore, letting $\tilde u \in \mathbb N$ be the point such that $G(\tilde u - 1) < F(\tilde s) \leq G(\tilde u)$, then this simplifies to
\begin{equation*}
    P(X = \tilde s, Y \leq y\, |\, X > 0, Y > 0) = \frac{1}{1 - p_2} \begin{cases}
        0 & \text{if } y = 0, \\
        G(y) - p_2 & \text{if } 0 < y < \tilde u, \\
        F(\tilde s) - p_2 & \text{if } y \geq \tilde u,
    \end{cases}
\end{equation*}
and so
\begin{equation*}
    P(X = \tilde s, Y = y\, |\, X > 0, Y > 0) = \frac{1}{1 - p_2} \begin{cases}
        G(y) - G(y - 1) & \text{if } 0 < y < \tilde u, \\
        F(\tilde s) - G(y - 1) & \text{if } y = \tilde u, \\
        0 & \text{otherwise.}
    \end{cases}
\end{equation*}
We now focus on deriving the distribution of $Y_{11}$. We can follow a similar argument as before to obtain
\begin{align*}
    P(Y_{11} \leq y) &= \frac{P(0 < Y \leq y, X > 0)}{P(X > 0, Y > 0)} \\
    &= \frac{1}{1 - p_2} \bigl(P(Y \leq y) - P(Y = 0) - P(Y \leq y, X = 0) + P(Y = 0, X = 0)\bigr) \\
    &= \frac{1}{1 - p_2} \bigl(G(y) - p_1\bigr).
\end{align*}
Hence, the probability mass function of $Y_{11}$ is
\begin{equation*}
    P(Y_{11}) = \frac{G(y) - G(y - 1)}{1 - p_2}
\end{equation*}
for $y \geq 1$ and zero otherwise.

We now have all the tools needed to derive $\rho_{S_{10}}$. We start by considering the probability of concordance and discordance separately. Note that, under $S_{10}$, $X_1$ is distributed as $X_{11}$ and $X_2$ is distributed as $X_{10}$. Thus, $P(X_1 < X_2\, |\, S_{10}) = 0$, and the probability of concordance can be given by
\begin{align*}
    P(X_1 > X_2, Y_1 > Y_3\, |\, S_{10}) &= P(Y_1 > Y_3\, |\, S_{10}) - P(X_1 \leq X_2, Y_1 > Y_3\, |\, S_{10}).
\end{align*}
In a similar fashion, we can express the probability of discordance as follows
\begin{align*}
    P(X_1 > X_2, Y_1 < Y_3\, |\, S_{10}) &= P(Y_1 < Y_3\, |\, S_{10}) - P(X_1 \leq X_2, Y_1 < Y_3\, |\, S_{10}).
\end{align*}
In these intermediate results, $Y_1$ and $Y_3$ have the same distribution and can therefore be interchanged. Substitution into the definition of $\rho_{S_{10}}$ together with $P(X_1 \leq X_2\, |\, S_{10}) = P(X_1 = X_2 = \tilde s\, |\, S_{10})$ yield
\begin{align*}
    \rho_{S_{10}} &= 3 P(X_1 \leq X_2, Y_1 < Y_3\, |\, S_{10}) - 3 P(X_1 \leq X_2, Y_1 > Y_3\, |\, S_{10}) \\
    &= 3 P(X_2 = \tilde s\, |\, S_{10}) \bigl[ P(X_1 = \tilde s, Y_1 < Y_3\, |\, S_{10}) - P(X_1 = \tilde s, Y_1 > Y_3\, |\, S_{10})\bigr] \\
    &= 3 P(X_{10} = \tilde s) \bigl[P(X_{11} = \tilde s) - 2 P(X_1 = \tilde s, Y_1 \geq Y_3\, |\, S_{10}) + P(X_1 = \tilde s, Y_1 = Y_3\, |\, S_{10})\bigr] \\
    &= 3 P(X_{10} = \tilde s) P(X_{11} = \tilde s) - 3 P(X_{10} = \tilde s) \sum_{y = 1}^\infty P(X_1 = \tilde s, Y_1 = y\, |\, S_{10}) \bigl[2 P(Y_{11} \leq y) - P(Y_{11} = y)\bigr].
\end{align*}
As we already have expressions of $P(X_{10} = \tilde s)$ and $P(X_{11} = \tilde s)$ available, we focus on the summation. We can substitute the margins and joint distributions into the formula, but it is necessary to be careful with the summation, since the closed form expression for $P(X = \tilde s, Y = y\, |\, X > 0, Y > 0)$ heavily depends on the value of $y$. We find that
\begin{align*}
    &\sum_{y = 1}^\infty P(X_1 = \tilde s, Y_1 = y\, |\, S_{10}) \bigl[2 P(Y_{11} \leq y) - P(Y_{11} = y)\bigr] \\
    &\qquad= \sum_{y = 1}^{\tilde u - 1} \frac{G(y) - G(y - 1)}{1 - p_2} \left[2 \frac{G(y) - p_2}{1 - p_2} - \frac{G(y) - G(y - 1)}{1 - p_2}\right] + \frac{F(\tilde s) - G(\tilde u - 1)}{1 - p_2} \left[2 \frac{G(\tilde u) - p_2}{1 - p_2} - \frac{G(\tilde u) - G(\tilde u - 1)}{1 - p_2}\right] \\
    &\qquad= \sum_{y = 1}^{\tilde u - 1} \frac{G^2(y) - G^2(y - 1) - 2 p_2 \bigl(G(y) - G(y - 1)\bigr)}{(1 - p_2)^2} + \frac{F(\tilde s) - G(\tilde u - 1)}{1 - p_2} \frac{G(\tilde u) + G(\tilde u - 1) - 2 p_2}{1 - p_2} \\
    &\qquad= \frac{G^2(\tilde u - 1) - p_2^2 - 2 p_2 \bigl[G(\tilde u - 1) - p_2\bigr]}{(1 - p_2)^2} + \frac{F(\tilde s) - G(\tilde u - 1)}{1 - p_2} \frac{G(\tilde u) + G(\tilde u - 1) - 2 p_2}{1 - p_2} \\
    &\qquad= \frac{p_2^2 - G(\tilde u) G(\tilde u - 1) + F(\tilde s) \bigl[G(\tilde u) + G(\tilde u - 1) - 2 p_2\bigr]}{(1 - p_2)^2}.
\end{align*}
We can now obtain the full formula for $\rho_{S_{10}}$ as follows
\begin{align*}
    p_{11}^2 p_{10} \rho_{S_{10}} &= 3 (1 - p_2) \bigl(F(\tilde s) - p_2\bigr) \bigl(p_2 - F(\tilde s-  1)\bigr) + \\
    &\qquad - 3 \bigl(p_2 - F(\tilde s - 1)\bigr) \left[p_2^2 - G(\tilde u)G(\tilde u - 1) + F(\tilde s) \bigl(G(\tilde u) + G(\tilde u - 1) - 2 p_2\bigr)\right].
\end{align*}
Under the assumption $p_1 \leq p_2$, we have $p_{01} = 0$, and therefore no expressions for $\rho_{S_{01}}$ and $\rho_{S_{00}}$ need to be derived. Since $p_1^* = P(X_{10} > X_{11}) = 0$, we have
\begin{align*}
    &3 p_{11} \left(p_{10} (1 - 2 p_1^* - p_1^\dagger) + p_{01} (1 - 2 p_2^* - p_2^\dagger)\right) + 3 (p_{11} p_{00} - p_{10} p_{01}) \\
    &\qquad= 3 (1 - p_2) (p_2 - p_1) - 3 \left(F(\tilde s) - p_2\right) \left(p_2 - F(\tilde s - 1)\right) + 3 (1 - p_2) p_1 \\
    &\qquad= 3 (1 - p_2) p_2 - 3 \left(F(\tilde s) - p_2\right) \left(p_2 - F(\tilde s - 1)\right).
\end{align*}
We now collect all the terms and obtain an upper bound of $\rho_S$ for discrete data as follows
\begin{align*}
    \rho_S &= (1 - p_2)^3 \rho_{S_{11}}^{\max} + 3 (1 - p_2) p_2 - 3 p_2 \bigl(F(\tilde s) - p_2\bigr) \bigl(p_2 - F(\tilde s - 1)\bigr) \\
    &\qquad + 3 \bigl(p_2 - F(\tilde s - 1)\bigr) \left[-p_2^2 + F(\tilde s) \bigl(2 p_2 - G(\tilde u) - G(\tilde u - 1)\bigr) + G(\tilde u) G(\tilde u - 1)\right] \\
    &= (1 - p_2)^3 \rho_{S_{11}}^{\max} + 3 (1 - p_2) p_2 \\
    &\qquad + 3 \bigl(p_2 - F(\tilde s - 1)\bigr) \left[F(\tilde s) \bigl(p_2 - G(\tilde u) - G(\tilde u - 1)\bigr) + G(\tilde u) G(\tilde u - 1)\right].
\end{align*}

\subsection{Lower bound}
To derive the attainable lower bound on $\rho_S$, we follow the same steps as in the previous section but consider the lower Fr\'echet-Hoeffding copula bound $W(u, v) = \max\{u + v - 1, 0\}$ instead. Under the assumption that $p_1 + p_2 \geq 1$, \cite{Denuit2017} noted that $p_{11} = 0$, $p_{10} = 1 - p_2$, $p_{01} = 1 - p_1$ and $p_{00} = p_1 + p_2 - 1$, hence the following
\begin{equation*}
    \rho_S^{\min} = - 3 (1 - p_1)(1 - p_2).
\end{equation*}
The derivations for the case where $p_1 + p_2 < 1$ are more elaborate and again require expressions for several margins and probabilities as in the previous section. this case involves many terms of the same form, therefore we start by recalling some notation. We define $w(x, y) = 1 - F(x) - G(y)$ for all $(x, y) \in \mathbb N^2$, observe its close relation to the lower Fr\'echet-Hoeffding bound $W$. From \cite{Denuit2017}, we have $p_{11} = 1 - p_1 - p_2 = w(0, 0)$, $p_{10} = p_2$, $p_{01} = p_1$ and $p_{00} = 0$. The probability mass functions of $X_{11}$ and $X_{10}$ derived by \cite{Perrone2023} can be used to define the distribution functions as
\begin{equation*}
    P(X_{11} = x) = \frac{1}{p_{11}} \begin{cases}
        F(x) - F(x - 1) & \text{if } x < \tilde s', \\
        1 - p_2 - F(x - 1) & \text{if } x = \tilde s', \\
        0 & \text{otherwise},
    \end{cases} \qquad P(X_{11} \leq x) = \frac{1}{p_{11}} \begin{cases}
        F(x) - p_1 & \text{if } x < \tilde s', \\
        1 & \text{if } x \geq \tilde s',
    \end{cases}
\end{equation*}
and
\begin{equation*}
    P(X_{10} = x) = \frac{1}{p_{10}} \begin{cases}
        0 & \text{if } x < \tilde s', \\
        F(x) + p_2 - 1 & \text{if } x = \tilde s', \\
        F(x) - F(x - 1) & \text{if } x > \tilde s',
    \end{cases} \qquad P(X_{10} \leq x) = \frac{1}{p_{10}} \begin{cases}
        0 & \text{if } x < \tilde s', \\
        F(x) + p_2 - 1 & \text{if } x \geq \tilde s',
    \end{cases}
\end{equation*}
for $x \geq 1$. Unfortunately, \cite{Perrone2023} do not specify the probability mass and distribution functions of $Y_{11}$ and $Y_{01}$, however, such distributions can be derived analogously and are given for $y \geq 1$ by
\begin{equation*}
    P(Y_{11} = y) = \frac{1}{p_{11}} \begin{cases}
        G(y) - G(y - 1) & \text{if } y < \tilde t', \\
        1 - p_1 - G(y - 1) & \text{if } y = \tilde t', \\
        0 & \text{otherwise,}
    \end{cases} \qquad P(Y_{11} \leq y) = \frac{1}{p_{11}} \begin{cases}
        G(y) - p_2 & \text{if } y < \tilde t', \\
        1 & \text{if } y \geq \tilde t',
    \end{cases}
\end{equation*}
and
\begin{equation*}
    P(Y_{01} = y) = \frac{1}{p_{01}} \begin{cases}
        0 & \text{if } y < \tilde t', \\
        G(y) + p_1 - 1 & \text{if } y = \tilde t', \\
        G(y) - G(y - 1) & \text{if } y > \tilde t',
    \end{cases} \qquad P(Y_{01} \leq y) = \frac{1}{p_{01}} \begin{cases}
        0 & \text{if } y < \tilde t', \\
        G(y) + p_1 - 1 & \text{if } y \geq \tilde t'.
    \end{cases}
\end{equation*}
Similarly to the upper bound, we have two points $\tilde s'$ and $\tilde t'$ that play a pivotal role in deriving the expressions for $\rho_S$ lower bound. Before actually computing the bounds of each component, we still consider the following probability
\begin{align*}
    &P(X = \tilde s',  Y \leq y\, |\, X > 0, Y > 0) \\
    &\qquad= \frac{P(X = \tilde s', 0 < Y \leq y)}{P(X > 0, Y > 0)} \\
    &\qquad= \frac{1}{p_{11}} \bigl(P(X \leq \tilde s', Y \leq y) - P(X \leq \tilde s' - 1, Y \leq y) - P(X \leq \tilde s',  Y = 0) + P(X \leq \tilde s' - 1, Y = 0)\bigr) \\
    &\qquad= \frac{1}{p_{11}} \bigl(\max\{F(\tilde s') + G(y) - 1, 0\} - \max\{F(\tilde s' - 1) + G(y) - 1, 0\} \\
    &\qquad\hspace*{2cm} - \max\{F(\tilde s') + p_2 - 1, 0\} + \max\{F(\tilde s' - 1) + p_2 - 1, 0\}\bigr) \\
    &\qquad= \frac{1}{p_{11}}\bigl(F(\tilde s') + G(y) - 1 - \max\{F(\tilde s' - 1) + G(y) - 1, 0\} - F(\tilde s') - p_2 + 1\bigr) \\
    &\qquad= \frac{1}{p_{11}} \bigl(G(y) - p_2 - \max\{F(\tilde s' - 1) + G(y) - 1, 0\}\bigr).
\end{align*}
Here, the fourth equality follows from $G(y) \geq G(0) = p_2$, and so $F(\tilde s') + G(y) - 1 \geq 0$. Recall that $u \in \mathbb N$ is the point such that $G(\tilde u' - 1) \leq 1 - F(\tilde s' - 1) < G(\tilde u')$, now we have
\begin{equation*}
    P(X = \tilde s', Y \leq y\, |\, X > 0, Y > 0) = \frac{1}{p_{11}} \begin{cases}
        G(y) - p_2 & \text{if } y < \tilde u', \\
        1 - F(\tilde s' - 1) - p_2 & \text{if } y \geq \tilde u',
    \end{cases}
\end{equation*}
and so
\begin{equation*}
    P(X = \tilde s', Y = y\, |\, X > 0, Y > 0) = \frac{1}{p_{11}} \begin{cases}
        G(y) - G(y - 1) & \text{if } y < \tilde u', \\
        1 - F(\tilde s' - 1) - G(y - 1) & \text{if } y = \tilde u', \\
        0 & \text{if } y > \tilde u'.
    \end{cases}
\end{equation*}
We can follow the same approach to derive an expression for $P(X = x, Y = \tilde t'\, |\, X > 0, Y > 0)$, but omit the details here.
In these probabilities, $Y$ is marginally distributed as $Y_{11}$. In light of the distribution of $Y_{11}$, it becomes clear that $\tilde u' \leq \tilde t'$, and by the same reasoning $\tilde v' \leq \tilde s'$. From here on, we need to make case distinctions between $\tilde u' = \tilde t'$, $\tilde v' = \tilde s'$, $\tilde u' > \tilde t'$ and $\tilde v' > \tilde s'$.

We now compute $\rho_{S_{10}}, \rho_{S_{01}}$ and $\rho_{S_{00}}$. First, we observe that $X_{11}$ and $X_{10}$ only have $\tilde s'$ in common in their supports. From the probability mass functions recalled earlier one obtains $P(X_{11} > X_{10}) = P(Y_{11} > Y_{01}) = 0$, and therefore the formulations for $\rho_{S_{10}}, \rho_{S_{01}}$ and $\rho_{S_{00}}$ can be simplified and tackled more easily. We consider each of these terms separately, starting with $\rho_{S_{10}}$. Here, $Y_1$ and $Y_2$ have the same distribution and therefore $P(Y_1 < Y_2\, |\, S_{10}) = P(Y_1 > Y_2\, |\, S_{10})$. Moreover, $X_2$ is distributed as $X_{10}$. Therefore, using that $P(X_1 > X_{10}\, |\, S_{10}) = 0$ yields
\begin{align*}
    \rho_{S_{10}} &= 3 P(X_1 < X_2, Y_1 < Y_3\, |\, S_{10}) - 3 P(X_1 < X_2, Y_1 > Y_3\, |\, S_{10}) \\
    &= 3 P(Y_1 < Y_3\, |\, S_{10}) - 3 P(X_1 \geq X_2, Y_1 < Y_3\, |\, S_{10}) - 3 P(Y_1 > Y_3\, |\, S_{10}) + 3 P(X_1 \geq X_2, Y_1 > Y_3\, |\ S_{10}) \\
    &= - 3 P(X_1 \geq X_2\, |\ S_{10}) + 6 P(X_1 \geq X_2, Y_1 \geq Y_3\, |\, S_{10}) - 3 P(X_1 \geq X_2, Y_1 = Y_3\, |\, S_{10}) \\
    &= - 3 P(X_{11} = \tilde s') P(X_{10} = \tilde s') + 3 P(X_{10} = \tilde s') \left[2 P(X_1 = \tilde s', Y_1 \geq Y_3\, |\, S_{10}) - P(X_1 = \tilde s', Y_1 = Y_3\, |\, S_{10})\right].
\end{align*}
The last equality follows from $P(X_{11} \geq X_{01}) = P(X_{11} = X_{10} = \tilde s')$. We now analyze the last term in parentheses. We sum over the support of the random variable $Y_1$ and substitute the relevant margins and distributions derived earlier. Although we have a closed form expression available for $P(X = \tilde s', Y = y\, |\, X > 0, Y > 0)$, they are different for $y \in \{1, \hdots, \tilde u' - 1\}$ and $x = \tilde u'$. Therefore, we consider again the two cases separately and write the terms as
\begin{align*}
    &2 P(X_1 = \tilde s', Y_1 \geq Y_3\, |\, S_{10}) - P(X_1 = \tilde s', Y_1 = Y_3\, |\, S_{10}) \\
    &\qquad= \sum_{x = 1}^{\tilde u' - 1} P(X = \tilde s', Y = y\, |\, X > 0, Y > 0) \Bigl(2 P(Y_{11} \leq y) - P(Y_{11} = y)\Bigr) \\
    &\qquad\qquad + P(X = \tilde s', Y = \tilde u'\, |\, X > 0, Y > 0) \Bigl(2 P(Y_{11} \leq \tilde u') - P(Y_{11} = \tilde v'\Bigr).
\end{align*}
As it is possible that $\tilde u' = 1$, the summation here should equal 0. in any case, $\tilde u' - 1 < \tilde t'$ and hence
\begin{align*}
    &\sum_{y = 1}^{\tilde u' - 1} P(X = \tilde s', Y = y\, |\, X > 0, Y > 0) \Bigl(2 P(Y_{11} \leq y) - P(Y_{11} = y)\Bigr) \\
    &\qquad = \sum_{y = 1}^{\tilde u' - 1} \frac{G(x) - G(x - 1)}{p_{11}} \left(2 \frac{G(x) - p_2}{p_{11}} - \frac{G(x) - G(x - 1)}{p_{11}}\right) \\
    &\qquad= - 2 p_2 \frac{G(\tilde u' - 1) - p_2}{p_{11}^2} + \frac{1}{p_{11}^2} \sum_{y = 1}^{\tilde u' - 1} \left(G^2(y) - G^2(y - 1)\right) \\
    &\qquad= \frac{\left(G(\tilde u' - 1) - p_2\right)^2}{p_{11}^2}.
\end{align*}
When $\tilde u' = 1$, one has $G(\tilde u' - 1) = F(0) = p_2$, and so this summation indeed equals zero. Therefore, we do not have to make another case distinction and may include this result entirely. It still remains to compute $P(X = \tilde s', Y = \tilde u'\, |\, X > 0, Y > 0)$. To do so, we need to consider two cases, one where $\tilde t' = \tilde u'$, and another where $\tilde t' > \tilde u'$. We first analyze $\tilde t' = \tilde u'$. Now $P(Y_{11} \leq \tilde u') = 1$, and hence
\begin{align*}
    &P(X = \tilde s', Y = \tilde u'\, |\, X > 0, Y > 0) \Bigl(2 P(Y_{11} \leq \tilde u') - P(Y_{11} = \tilde u')\Bigr) \\
    &\qquad= \frac{1 - F(\tilde s' - 1) - G(\tilde u' - 1)}{p_{11}} \left(2 - \frac{1 - p_1 - G(\tilde u' - 1)}{p_{11}}\right) \\
    &\qquad= \frac{1 - F(\tilde s' - 1) - G(\tilde u' - 1)}{p_{11}} \frac{(1 - p_1) + G(\tilde u' - 1) - 2 p_2}{p_{11}}.
\end{align*}
On the other hand, when $\tilde t' > \tilde u'$, one obtains
\begin{align*}
    &P(X = \tilde s', Y = \tilde u'\, |\, X > 0, Y > 0) \Bigl(2 P(Y_{11} \leq \tilde u') - P(Y_{11} = \tilde u')\Bigr) \\
    &\qquad= \frac{1 - F(\tilde s' - 1) - G(\tilde u' - 1)}{p_{11}} \left(2 \frac{G(\tilde u') - p_2}{p_{11}} - \frac{G(\tilde u') - G(\tilde u' - 1)}{p_{11}}\right) \\
    &\qquad= \frac{1 - F(\tilde s' - 1) - G(\tilde u' - 1)}{p_{11}} \frac{G(\tilde u')  + G(\tilde u' - 1) - 2 p_2}{p_{11}}.
\end{align*}
Now letting $I_2 = \min\{1 - p_1, G(\tilde u')$, one has
\begin{equation*}
    P(X = \tilde s', Y = \tilde u'\, |\, X > 0, Y > 0) \Bigl(2 P(Y_{11} \leq \tilde u') - P(Y_{11} = \tilde u')\Bigr) = \frac{\bigl(1 - F(\tilde s' -  1) - G(\tilde u' - 1)\bigr) \bigl(I_1 + G(\tilde u' - 1) - 2 p_2\bigr)}{p_{11}^2}
\end{equation*}
Finally, substituting the relevant results and margins into the definition of $\rho_{S_{10}}$ yields
\begin{align*}
    p_{11}^2 p_{10} \rho_{S_{10}} &= 3 w(0, 0) w(\tilde s', 0) w(\tilde s'-,0) - 3 w(\tilde s', 0) \left[\Bigl(G(\tilde u' - 1) - p_2\Bigr)^2 + w(\tilde s'-1, \tilde u'-1) \Bigl(I_1 + G(\tilde u' - 1) - 2 p_2\Bigr)\right].
\end{align*}
In a similar manner, one may show that
\begin{align*}
    p_{11}^2 p_{01} \rho_{S_{01}} = 3 w(0, 0) w(0, \tilde t') w(0, \tilde t'-1) - 3 w(0, \tilde t') \left[\Bigl(F(\tilde v' - 1) - p_1\Bigr)^2 + w(\tilde v'-, \tilde t'-1) \Bigl(I_2 + F(\tilde v' - 1) - 2 p_1\Bigr)\right].
\end{align*}
To finalize the proof, we obtain $\rho_{S_{00}}$. Since under this condition, only $X_1 < X_2$ and $Y_1 < Y_3$ hold with some probability, the term simplifies to
\begin{align*}
    \rho_{S_{00}} &= 3 P(X_1 < X_2, Y_1 < Y_3\, |\, S_{00}) \\
    &= 3 \Bigl(1 - P(Y_1 \geq Y_3\, |\, S_{00}) - P(X_1 \geq X_2\, |\, S_{00}) + P(X_1 \geq X_2, Y_1 \geq Y_3\, |\, S_{00})\Bigr) \\
    &= 3 \Bigl(1 - P(Y_{11} = \tilde t') P(Y_{01} = \tilde t') - P(X_{11} = \tilde s') P(X_{10} = \tilde s') \\
    &\qquad\qquad + P(X = \tilde s', Y = \tilde t'\, |\, X > 0, Y > 0) P(X_{10} = \tilde s') P(Y_{01} = \tilde t')\Bigr) \\
    &= 3 \left(1 + \frac{W(0, \tilde t'-1) W(0, \tilde t')}{p_{11} p_{01}} + \frac{W(\tilde s'-1, 0) W(\tilde s', 0)}{p_{11} p_{10}} + \frac{W(\tilde s'-, \tilde t'-1) W(\tilde s',0) W(0,\tilde t')}{p_{11} p_{10} p_{01}} \mathbbm 1(\tilde s' = \tilde v')\right).
\end{align*}
For the remaining terms in Eq.~(6), we can substitute the results on $p_1^*, p_1^\dagger, p_2^*, p_2^\dagger$ from \cite{Perrone2023} and find
\begin{align*}
    &3 p_{11} \left(p_{10} (1 - 2 p_1^* - p_1^\dagger) + p_{01} (1 - 2 p_2^* - p_2^\dagger)\right) \\
    &\qquad = - w(0, 0) (p_1 + p_2) - 3 w(\tilde s', 0) w(\tilde s'-, 0) - 3 w(0, \tilde t') w(0, \tilde t'-1) - 3 p_1 p_2.
\end{align*}
We can collect all of these results and find the following expression for $\rho_S^{\min}$, which concludes the proof
\begin{align*}
    \rho_S^{\min} &= w(0, 0)^3 \rho_{S_{11}}^{\min} + 3 w(0, 0) (p_1 p_2 - p_1 - p_2) - 3 p_1 p_2 + w(\tilde s'-, \tilde u'-1) w(\tilde s', 0) w(0, \tilde t') \mathbbm 1(\tilde s' = \tilde v') + \\
            &\quad - 3 w(\tilde s', 0) \left[p_2 w(\tilde s'-, 0) + \bigl(G(\tilde u'-1) - p_2\bigr)^2 + w(\tilde s'-, \tilde u'-1) \bigl(I_2 - 2 p_2 + G(\tilde u'-1)\bigr)\right] \\
            &\quad - 3 w(0, \tilde t') \left[p_1 w(0, \tilde t'-1) + \bigl(F(\tilde v'-1) - p_1\bigr)^2 + w(\tilde v'-, \tilde t'-1) \bigl(I_1 - 2 p_1 + F(\tilde v'-1)\bigr)\right].
\end{align*} \qed

\section{Simulation results}\label{app:sim}

\begin{figure}[h!]
    \centering
    \includegraphics[width=0.25\linewidth]{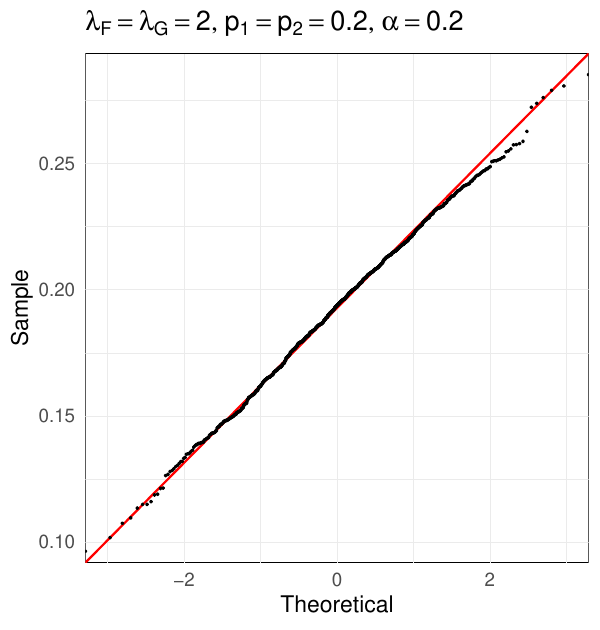}
    \includegraphics[width=0.25\linewidth]{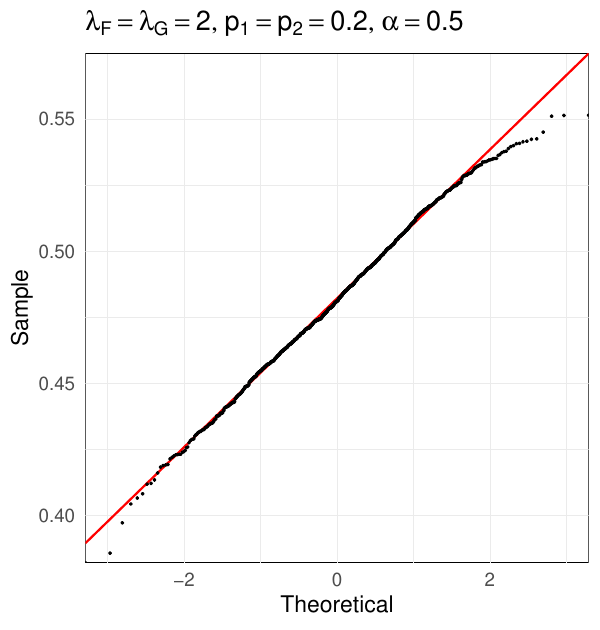}
    \includegraphics[width=0.25\linewidth]{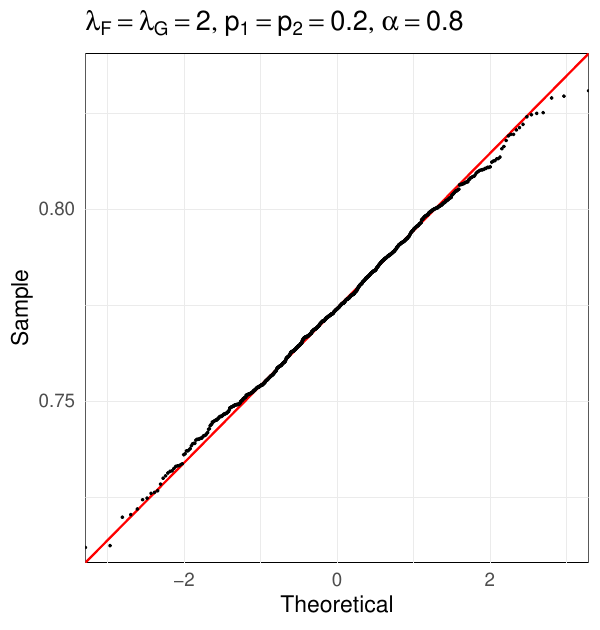}
    
    \includegraphics[width=0.25\linewidth]{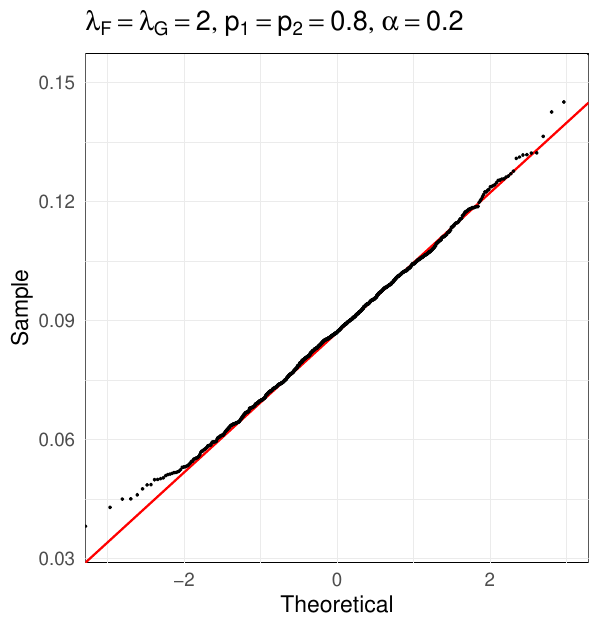}
    \includegraphics[width=0.25\linewidth]{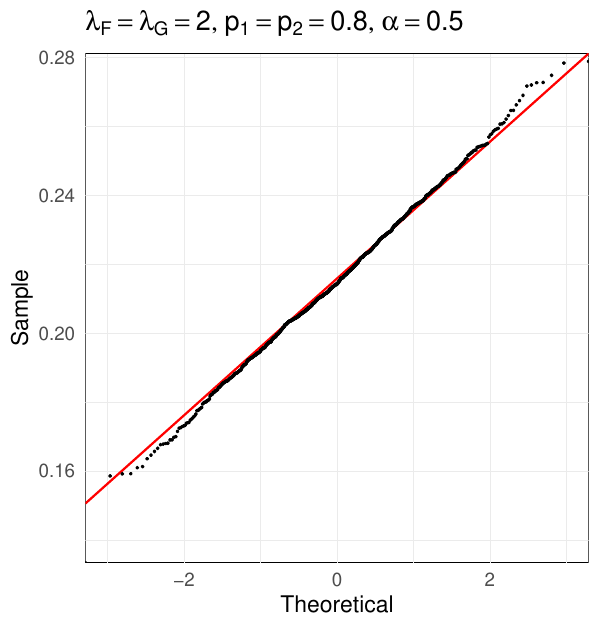}
    \includegraphics[width=0.25\linewidth]{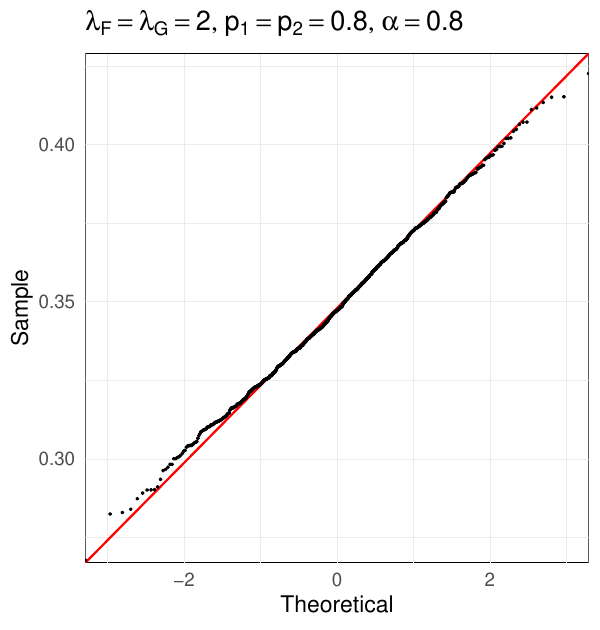}
    \caption{QQ-plots of $\hat \rho_A$ with sample size $N = 1000$ under various parameter configurations. The simulations are based on 1000 repetitions.}
\end{figure}

\begin{table}[h!]
\centering
\caption{Simulation results based on 1000 repetitions. The reported MSE* is 100 times the estimated MSE.}
\scriptsize
\begin{tabular}{lllllllllllllllllll}
\hline \hline
      &      & \multicolumn{5}{l}{$\lambda_F =   \lambda_G = 2, p_1 = p_2 = 0.2$}         &                 &                 &      &  & \multicolumn{5}{l}{$\lambda_F =   \lambda_G = 2, p_1 = p_2 = 0.8$}         &                 &                 &      \\ \cline{3-10} \cline{12-19} 
      &      & $\rho_S$        & \multicolumn{2}{l}{$\hat \rho_A$} &             &        & \multicolumn{2}{l}{$\hat \rho_N$} &      &  & $\rho_S$        & \multicolumn{2}{l}{$\hat \rho_A$} &             &        & \multicolumn{2}{l}{$\hat \rho_N$} &      \\ \cline{13-15} \cline{17-19} 
$\alpha$ & N    &                 & Mean            & MSE*             & SE          &        & Mean            & MSE*             & SE   &  &                 & Mean            & MSE*             & SE          &        & Mean            & MSE*             & SE   \\ \cline{1-6} \cline{8-10} \cline{12-15} \cline{17-19} 
0.2   & 100  & 0.19            & 0.19            & 0.93            & 0.10        &        & 0.19            & 0.91            & 0.10 &  & 0.09            & 0.09            & 0.28            & 0.05        &        & 0.09            & 0.28            & 0.05 \\
      & 150  &                 & 0.19            & 0.59            & 0.08        &        & 0.19            & 0.58            & 0.08 &  &                 & 0.09            & 0.19            & 0.04        &        & 0.09            & 0.19            & 0.04 \\
      & 250  &                 & 0.19            & 0.38            & 0.06        &        & 0.19            & 0.38            & 0.06 &  &                 & 0.09            & 0.12            & 0.03        &        & 0.09            & 0.12            & 0.03 \\
      & 500  &                 & 0.19            & 0.18            & 0.04        &        & 0.19            & 0.18            & 0.04 &  &                 & 0.09            & 0.06            & 0.02        &        & 0.09            & 0.06            & 0.02 \\
      & 1000 &                 & 0.19            & 0.09            & 0.03        &        & 0.19            & 0.09            & 0.03 &  &                 & 0.09            & 0.03            & 0.02        &        & 0.09            & 0.03            & 0.02 \\ \cline{1-6} \cline{8-10} \cline{12-15} \cline{17-19} 
0.8   & 100  & 0.76            & 0.77            & 0.41            & 0.06        &        & 0.75            & 0.38            & 0.06 &  & 0.35            & 0.34            & 0.56            & 0.07        &        & 0.34            & 0.56            & 0.07 \\
      & 150  &                 & 0.77            & 0.29            & 0.05        &        & 0.75            & 0.26            & 0.05 &  &                 & 0.35            & 0.35            & 0.06        &        & 0.35            & 0.35            & 0.06 \\
      & 250  &                 & 0.77            & 0.19            & 0.04        &        & 0.76            & 0.16            & 0.04 &  &                 & 0.35            & 0.22            & 0.05        &        & 0.34            & 0.22            & 0.05 \\
      & 500  &                 & 0.77            & 0.11            & 0.03        &        & 0.76            & 0.08            & 0.03 &  &                 & 0.35            & 0.11            & 0.03        &        & 0.35            & 0.11            & 0.03 \\
      & 1000 &                 & 0.77            & 0.07            & 0.02        &        & 0.76            & 0.04            & 0.02 &  &                 & 0.35            & 0.06            & 0.02        &        & 0.35            & 0.06            & 0.02 \\ \hline \hline
      &      & \multicolumn{5}{l}{$\lambda_F = 2, \lambda_G = 8, p_1 =   p_2 = 0.2$}      &                 &                 &      &  & \multicolumn{5}{l}{$\lambda_F = 2, \lambda_G = 8, p_1 =   p_2 = 0.8$}      &                 &                 &      \\ \cline{3-10} \cline{12-19} 
      &      & $\rho_S$        & \multicolumn{2}{l}{$\hat \rho_A$} &             &        & \multicolumn{2}{l}{$\hat \rho_N$} &      &  & $\rho_S$        & \multicolumn{2}{l}{$\hat \rho_A$} &             &        & \multicolumn{2}{l}{$\hat \rho_N$} &      \\ \cline{4-6} \cline{8-10} \cline{13-15} \cline{17-19} 
$\alpha$ & N    &                 & Mean            & MSE*             & SE          &        & Mean            & MSE*             & SE   &  &                 & Mean            & MSE*             & SE          &        & Mean            & MSE*             & SE   \\ \cline{1-6} \cline{8-10} \cline{12-15} \cline{17-19} 
0.2   & 100  & 0.19            & 0.19            & 0.93            & 0.10        &        & 0.19            & 0.92            & 0.10 &  & 0.09            & 0.09            & 0.30            & 0.05        &        & 0.09            & 0.30            & 0.05 \\
      & 150  &                 & 0.19            & 0.69            & 0.08        &        & 0.19            & 0.68            & 0.08 &  &                 & 0.09            & 0.20            & 0.04        &        & 0.09            & 0.20            & 0.04 \\
      & 250  &                 & 0.19            & 0.42            & 0.06        &        & 0.18            & 0.41            & 0.06 &  &                 & 0.09            & 0.12            & 0.03        &        & 0.09            & 0.12            & 0.03 \\
      & 500  &                 & 0.19            & 0.19            & 0.04        &        & 0.19            & 0.19            & 0.04 &  &                 & 0.09            & 0.06            & 0.03        &        & 0.09            & 0.06            & 0.03 \\
      & 1000 &                 & 0.19            & 0.10            & 0.03        &        & 0.19            & 0.09            & 0.03 &  &                 & 0.09            & 0.03            & 0.02        &        & 0.09            & 0.03            & 0.02 \\ \cline{1-6} \cline{8-10} \cline{12-15} \cline{17-19} 
0.8   & 100  & 0.75            & 0.76            & 0.40            & 0.06        &        & 0.75            & 0.39            & 0.06 &  & 0.35            & 0.34            & 0.52            & 0.07        &        & 0.34            & 0.52            & 0.07 \\
      & 150  &                 & 0.76            & 0.28            & 0.05        &        & 0.75            & 0.27            & 0.05 &  &                 & 0.35            & 0.38            & 0.06        &        & 0.35            & 0.38            & 0.06 \\
      & 250  &                 & 0.76            & 0.17            & 0.04        &        & 0.75            & 0.16            & 0.04 &  &                 & 0.35            & 0.22            & 0.05        &        & 0.35            & 0.22            & 0.05 \\
      & 500  &                 & 0.76            & 0.09            & 0.03        &        & 0.75            & 0.08            & 0.03 &  &                 & 0.35            & 0.11            & 0.03        &        & 0.35            & 0.11            & 0.03 \\
      & 1000 &                 & 0.76            & 0.05            & 0.02        &        & 0.75            & 0.04            & 0.02 &  &                 & 0.35            & 0.05            & 0.02        &        & 0.35            & 0.05            & 0.02 \\ \hline \hline
      &      & \multicolumn{5}{l}{$\lambda_F = \lambda_G = 8, p_1 =   p_2 = 0.2$}         &                 &                 &      &  & \multicolumn{5}{l}{$\lambda_F = \lambda_G = 8, p_1 =   p_2 = 0.8$}         &                 &                 &      \\ \cline{3-10} \cline{12-19} 
      &      & $\rho_S$        & \multicolumn{2}{l}{$\hat \rho_A$} &             &        & \multicolumn{2}{l}{$\hat \rho_N$} &      &  & $\rho_S$        & \multicolumn{2}{l}{$\hat \rho_A$} &             &        & \multicolumn{2}{l}{$\hat \rho_N$} &      \\ \cline{4-6} \cline{8-10} \cline{13-15} \cline{17-19} 
$\alpha$ & N    &                 & Mean            & MSE*             & SE          &        & Mean            & MSE*             & SE   &  &                 & Mean            & MSE*             & SE          &        & Mean            & MSE*             & SE   \\ \cline{1-6} \cline{8-10} \cline{12-15} \cline{17-19} 
0.2   & 100  & 0.20            & 0.20            & 1.07            & 0.10        &        & 0.20            & 1.06            & 0.10 &  & 0.10            & 0.10            & 0.35            & 0.06        &        & 0.10            & 0.35            & 0.06 \\
      & 150  &                 & 0.20            & 0.71            & 0.08        &        & 0.20            & 0.71            & 0.08 &  &                 & 0.10            & 0.22            & 0.05        &        & 0.10            & 0.22            & 0.05 \\
      & 250  &                 & 0.19            & 0.44            & 0.07        &        & 0.19            & 0.44            & 0.07 &  &                 & 0.10            & 0.13            & 0.04        &        & 0.10            & 0.13            & 0.04 \\
      & 500  &                 & 0.20            & 0.19            & 0.04        &        & 0.19            & 0.19            & 0.04 &  &                 & 0.10            & 0.07            & 0.03        &        & 0.10            & 0.07            & 0.03 \\
      & 1000 &                 & 0.20            & 0.10            & 0.03        &        & 0.20            & 0.10            & 0.03 &  &                 & 0.10            & 0.03            & 0.02        &        & 0.10            & 0.03            & 0.02 \\ \cline{1-6} \cline{8-10} \cline{12-15} \cline{17-19} 
0.8   & 100  & 0.79            & 0.79            & 0.41            & 0.06        &        & 0.78            & 0.41            & 0.06 &  & 0.39            & 0.38            & 0.57            & 0.07        &        & 0.38            & 0.57            & 0.07 \\
      & 150  &                 & 0.79            & 0.27            & 0.05        &        & 0.78            & 0.27            & 0.05 &  &                 & 0.39            & 0.38            & 0.06        &        & 0.39            & 0.38            & 0.06 \\
      & 250  &                 & 0.79            & 0.17            & 0.04        &        & 0.79            & 0.17            & 0.04 &  &                 & 0.39            & 0.21            & 0.05        &        & 0.39            & 0.21            & 0.05 \\
      & 500  &                 & 0.79            & 0.09            & 0.03        &        & 0.79            & 0.09            & 0.03 &  &                 & 0.39            & 0.11            & 0.03        &        & 0.39            & 0.11            & 0.03 \\
      & 1000 &                 & 0.79            & 0.04            & 0.02        &        & 0.79            & 0.04            & 0.02 &  &                 & 0.39            & 0.05            & 0.02        &        & 0.39            & 0.05            & 0.02 \\ \hline \hline
\end{tabular}
\end{table}

\begin{table}[h!]
\centering
\footnotesize
\caption{Estimates of the standard deviation of $\hat \rho_S$ for $N = 150$ based on Monte-Carlo simulations with 1000 repetitions ($\hat \sigma_{\rm MC}$) and the average estimated standard deviation using 1000 bootstrap resamples ($\hat \sigma_{\rm B}$) over 150 simulations.}
{
\begin{tabular}{lllllllllll} \hline\hline
&&& \multicolumn{2}{c}{$\lambda_F = \lambda_G = 2$} && \multicolumn{2}{l}{$\lambda_F = 2, \lambda_G = 8$} && \multicolumn{2}{l}{$\lambda_F = \lambda_G = 8$} \\\cline{4-5} \cline{7-8} \cline{10-11}
$p_1 = p_2$ & $\alpha$ && $\hat \sigma_{\rm MC}$ &  $\hat \sigma_{\rm B}$ && $\hat \sigma_{\rm MC}$ & $\hat \sigma_{\rm B}$ && $\hat \sigma_{\rm MC}$ & $\hat \sigma_{\rm B}$ \\\hline
0.2 & 0.2 &  & 0.077 & 0.080 &  & 0.083 & 0.081 &  & 0.085 & 0.083 \\
    & 0.5 &  & 0.074 & 0.072 &  & 0.076 & 0.073 &  & 0.077 & 0.076 \\
    & 0.8 &  & 0.052 & 0.051 &  & 0.052 & 0.050 &  & 0.052 & 0.052 \\\hline
0.8 & 0.2 &  & 0.044 & 0.043 &  & 0.045 & 0.046 &  & 0.047 & 0.047 \\
    & 0.5 &  & 0.054 & 0.053 &  & 0.056 & 0.054 &  & 0.058 & 0.055 \\
    & 0.8 &  & 0.059 & 0.059 &  & 0.062 & 0.060 &  & 0.062 & 0.060 \\\hline\hline
\end{tabular}}
\end{table}

\begin{table}[h!]
\caption{Average of the bounds' estimates for $N = 150$, based on 1000 repetitions. }
\footnotesize
\centering
\begin{tabular}{lllllllllll}
\hline \hline
$\lambda_F$ & $\lambda_G$ & $p_1$  & $p_2$  & True bounds     & Estimated bounds &  & $p_1$  & $p_2$  & True bounds     & Estimated bounds \\ \hline
2       & 2       & 0.2 & 0.2 & $[-0.89, 0.95]$ & $[-0.91, 0.95]$  &  & 0.8 & 0.2 & $[-0.36, 0.42]$ & $[-0.36, 0.42]$  \\
        &         &     & 0.5 & $[-0.76, 0.79]$ & $[-0.76, 0.80]$  &  &     & 0.5 & $[-0.22, 0.43]$ & $[-0.22, 0.42]$  \\
        &         &     & 0.8 & $[-0.36, 0.42]$ & $[-0.36, 0.42]$  &  &     & 0.8 & $[-0.09, 0.43]$ & $[-0.09, 0.40]$  \\ \hline
2       & 8       & 0.2 & 0.2 & $[-0.93, 0.94]$ & $[-0.95, 0.97]$  &  & 0.8 & 0.2 & $[-0.41, 0.43]$ & $[-0.41, 0.43]$  \\
        &         &     & 0.5 & $[-0.83, 0.85]$ & $[-0.83, 0.86]$  &  &     & 0.5 & $[-0.26, 0.43]$ & $[-0.26, 0.43]$  \\
        &         &     & 0.8 & $[-0.41, 0.47]$ & $[-0.42, 0.48]$  &  &     & 0.8 & $[-0.10, 0.43]$ & $[-0.26, 0.43]$  \\ \hline
8       & 8       & 0.2 & 0.2 & $[-0.97, 0.99]$ & $[-0.98, 0.99]$  &  & 0.8 & 0.2 & $[-0.48, 0.48]$ & $[-0.47, 0.48]$  \\
        &         &     & 0.5 & $[-0.86, 0.86]$ & $[-0.86, 0.87]$  &  &     & 0.5 & $[-0.30, 0.49]$ & $[-0.38, 0.48]$  \\
        &         &     & 0.8 & $[-0.48, 0.48]$ & $[-0.47, 0.48]$  &  &     & 0.8 & $[-0.12, 0.49]$ & $[-0.12, 0.45]$  \\ \hline \hline
\end{tabular}
\end{table}

\end{document}